\documentclass{amsart}
 \newtheorem{thm}{Theorem}[section]
 \newtheorem{cor}[thm]{Corollary}
 \newtheorem{lem}[thm]{Lemma}
 \newtheorem{prop}[thm]{Proposition}
 \theoremstyle{definition}
 
 \theoremstyle{remark}

 \numberwithin{equation}{section}

\usepackage{bm}
\usepackage{graphicx}
\usepackage{amssymb}
\usepackage{hyperref}

\begin{document}

\title[Planar elliptic growth]
 {Planar Elliptic Growth}
\author[Khavinson]{Dmitry Khavinson}

\address{Department of Mathematics and Statistics,
 University of South Florida, Tampa, FL 33620}

\email{dkhavins@cas.usf.edu}
\thanks{This work was supported by the 20070483ER project
Minimal Description of Complex Interfaces of the LDRD program at
LANL. The first and third authors were also partially supported by
the NSF grants.}
\author[Mineev]{Mark Mineev-Weinstein}
\address{LANL,
MS-365, Los Alamos, NM 87545}

\email{mariner@lanl.gov}

\author[Putinar]{Mihai Putinar}
\address{Department of Mathematics,
University of California, Santa Barbara, CA 93106}

\email{mputinar@math.ucsb.edu}
\subjclass{Primary 76S05  ; Secondary 76D27, 31A25, 30C20, 31B35,
35J10 }

\keywords{Moving boundaries, Elliptic growth, Laplacian growth,
Schwarz function, Beltrami equation, Schr\"odinger operator,
Dirichlet problem, Carleman equation}

\date{today}
\dedicatory{To our friend Bj\"orn Gustafsson}

\begin{abstract}

 The planar elliptic extension of the Laplacian growth is, after a proper
parametrization, given in a form of a solution to the equation for
area-preserving diffeomorphisms.
The infinite set of conservation laws associated with such elliptic
growth is interpreted in terms of potential theory, and the
relations between two major forms of the elliptic growth are
analyzed.  The constants of integration for closed form solutions
are identified as the singularities of the Schwarz function, which
are located both inside and outside the moving contour.
Well-posedness of the recovery of the elliptic operator governing
the process from the continuum of interfaces parametrized by time is
addressed and two examples of exact solutions of elliptic growth are
presented.

\end{abstract}

\maketitle

\section{Introduction}

Several moving boundary processes, such as solidification
\cite{Langer}, electrodeposition \cite{Gollub}, viscous fingering
\cite{Bensimon}, and bacterial growth \cite{BenJacob}, to name a
few, can be reduced, after some idealizations, to the Laplacian
growth, which can be described  as follows:

\begin{equation}\label{growthEq}
V(\xi) = \partial_n G_{D(t)}(\xi,a).
\end{equation}

Here $V$ is the normal component of the velocity of the boundary
$\partial D(t)$ of the moving domain $D(t) \subset {\mathbb R}^d$,
$\xi \in
\partial D(t)$, $t$ is time, $\partial_n$ is the normal component of
the gradient, and $G_{D(t)}(\xi,a)$ is the Green function of the
domain $D(t)$ for the Laplace operator with a unit source located at
the point $a \in D(t)$.

In two dimensions this equation can be rewritten as the
area-preserving diffeomorphism identity
\begin{equation}\label{areapres}
\Im \,(\bar z_t z_{\phi}) = 1,
\end{equation}
where $z(t,\phi): = \partial D(t)$ is the moving boundary
parameterized by $\phi \in [0,\,2\pi]$ and conformal when
analytically extended in the region  $\Im \phi \leq 0$ \cite{Galin,
Polubarinova-Kochina}.  The equation (\ref{areapres}) possesses many
remarkable properties, among which, the most noticeable ones
are the existence of an infinite set of conservation laws:
\begin{equation}\label{moments}
C_n = \int_{D(t)}\, z^n \, dx\,dy,
\end{equation}
where $n$ runs over all non-negative \cite{Richardson1972}
(non-positive \cite{Mineev1990}) integers in the case of a finite
(infinite) domain $D(t)$, and an impressive list of exact
time-dependent closed form solutions \cite{Varchenko}. For a
beautiful interpretation of conserved quantities $C_n$ as
coefficients of the multi-pole expansion of the fictitious
Newtonian potential created by matter uniformly occupying the
domain $D(t)$ see, e.g., \cite{Varchenko}.

It was established in \cite{Mineev-Weinstein2000} that the interface dynamics
described by  (\ref{areapres}) is equivalent to the dispersionless integrable
2D Toda hierarchy \cite{2DToda}, constrained by the string equation. Remarkably,
this hierarchy, being one of the richest existing integrable structures,
describes an existing theory of 2D quantum gravity (see the comprehensive review
\cite{2DToda} and references therein).  The work \cite{Mineev-Weinstein2000}
generated a splash of activity in apparently different mathematical and physical
directions revealing profound connections between Laplacian growth and random
matrices \cite{Krichever2001}, the Whitham theory \cite{Krichever2004}, and
quadrature domains \cite{Mineev-Put-Teo}.

In this paper we present a natural extension of the Laplacian
growth, where the Green function of $D(t)$ for the Laplace operator
$\nabla^2$ in the RHS of (1.1) is replaced by the Green function of
a linear elliptic operator,
\begin{equation}\label{Loperator}
L = \nabla \cdot (\lambda({\bf x})\nabla) - u({\bf x}), \qquad
\lambda({\bf x}) > 0, \qquad {\bf x} \in {\mathbb R}^d.
\end{equation}
Such a process, which is natural to be named {\it an elliptic
growth}, is clearly much more common in physics than the Laplacian
growth.

Consider, for instance, viscous fingering between viscous and
inviscid fluids in the porous media governed by Darcy's law
\begin{equation}\label{Darcy}
{\bf v} = - \lambda \nabla p,
\end{equation}
where $\lambda$ is the filtration coefficient of the media and $p$
is the pressure (equal to the Green function, $G_{D(t)}$, defined in
(1.1) in most of the cases of interest for us). One can easily
imagine a non-homogeneous media where the filtration coefficient
$\lambda$ is space-dependent.  Such examples of elliptic growth,
where the elliptic operator $L$ has the form of the Laplace-Beltrami
operator, $L = \nabla \cdot \lambda \nabla$, and $\lambda$ is a
prescribed function of ${\bf x}$, will be called an elliptic growth
{\it of the Beltrami type}.  It is clear that all moving boundary
problems other than viscous fingering with a non-homogeneous kinetic
coefficient $\lambda$
fall into this category.

From a mathematical point of view this process is the Laplacian
growth occurring on curved surfaces instead of the Euclidean
plane. In this case the Laplace equation is naturally replaced by
the Laplace-Beltrami equation, and $\lambda$ (that can be a matrix
instead of a scalar as it is in our case) is related to the metric
tensor. There are several works addressing the Hele-Shaw problem
on curved surfaces and we will mention below those few related to
the integrable mathematical structure of elliptic growth.

Another major source of examples of elliptic growth is related to
screening effects, when $u \neq 0$, while $\lambda$ is constant in
(\ref{Loperator}).  The simplest example of this kind is an
electrodeposition, where the field $p$ is the electrostatic
potential of the electrolyte.  It is known that in reality
electrolytes ions are always locally surrounded by a cloud of
oppositely charged ions. This screening modifies the Laplace
equation for the electrostatic potential by adding to the Laplace
operator the negative screening term, $-u(x)$, which stands for the
inverse square of the radius of the Debye-Hukkel screening in the
classical plasma \cite{Landau}.  For the homogeneous screening $u$
is a (positive) constant, so the operator $L$ becomes the Helmholtz
operator, while for the non-homogeneous case, when $u$ is not a
constant, $L$ is a standard Schr\"odinger operator.  Motivated by
this example,  we will call the moving boundary problem for $L =
\nabla^2 - u$ an elliptic growth of {\it Schr\"odinger type}.

We show that  these rather general types of elliptic growth still
retain remarkable mathematical properties, similar to those
possessed by the Laplacian growth.  A mixed case with a non-constant
$\lambda$ and non-zero $u$ also shares similar properties but is
less representative in physics and can always be reduced to one of
the two former types of elliptic growth by a simple transformation
described later on in the article.  For completeness  we shall
indicate another class of  elliptic growth when the fluid density
$\rho$ changes in space while it is constant in time.  This happens
for instance when porosity (fraction of porous media accessible for
fluid) is space-dependent. In this case the continuity equation for
incompressible fluid in porous media has the form
\begin{equation}\label{phi-lambda}
\nabla(\rho({\bf x})\lambda({\bf x})\nabla p) = 0,
\end{equation}
while (\ref{Darcy}) still holds.  This case presents an additional
extension of the elliptic growth related to  potential theory with a
non-uniform density, as will be shown below.

In prior works on elliptic growth an infinite number of conservation
laws, regarded as extensions of (\ref{moments}), were identified in
\cite{Varchenko, Mineev-Weinstein1993}.  Also an integrable example
in 2D, which corresponds to a very special choice of the
conductivity function, $\lambda({\bf x})$, was explicitly
constructed in \cite{Lou}-\cite{LY}.  The elliptic growth in these
works was reduced to the well-known  Calogero-Moser integrable
system.

The present article contains several new results on elliptic growth
and reviews the known conservation laws from a slightly novel
perspective. It is organized as follows:

- In Section 2 we interpret (\ref{areapres}) as the equation of the
area-preserved diffeomorphism in 2D and analyze its connections with
the Laplacian growth.

- Section 3 contains the definition of elliptic growth, the
conservation laws for this process and recasts the latter in terms
of the inverse non-Newtonian potential theory.

- In Section 4 we obtain the equation (\ref{areapres}) for elliptic
growth of the Beltrami type by introducing the function $q$,
conjugate with respect to $p$ defined in (\ref{Loperator}), which
plays the role of a stream function for the incompressible fluid;
furthermore, we obtain the Beltrami equation for the function
$p+iq$.

- Section 5 is devoted to analyzing connections between the elliptic
growth of the Beltrami and  Schr\"odinger types, elucidating the
difficulties of parametrization of the interface for the
Schr\"odinger type.

- In Section 6 we reformulate the elliptic growth in terms of the
Schwarz function of the moving interface.

- Section 7 addresses a well-posedness of a recovery problem for the
operator $L$ from the continuum of moving interfaces parameterized
by time.

- In Section 8 we discuss  Herglotz' theorem as the main device to
generate exact solutions and present two examples of the exact
closed form solutions of the elliptic growth.  We also identify the
constants of motion of these solutions as the singularities of the
Schwarz function of the moving contour.

- Section 9 contains brief conclusions.

Due to the fact that one of us is labelled as a theoretical
physicist and following the customs of the physics community the
references do not appear in alphabetical order.

{\bf Acknowledgement.} This work took shape during a visit of the
first and third author to the Los Alamos
National Laboratory. They warmly thank this institution for an inspiring and stimulating atmosphere.\\

\subsection{List of notations and conventions} We collect below a few basic definitions and
notations used throughout the text.
\bigskip


$\nabla^2 = \Delta,$ \qquad $\nabla f = {\rm grad} \, f$;

$\nabla (\bf U) = \nabla \cdot {\bf U} = {\rm div} \, {\bf U}$,
where ${\bf U}$ is a vector field;


$\mathcal C^\omega$ denotes the class of real analytic functions;

$\dot{h} = \frac{\partial h}{\partial t}$;

$n$ stands for the outer unit normal to the moving boundary $\Gamma
= \Gamma(t)$;

$\ell$ denotes the arc length on the boundary $\Gamma$;

$dA = \frac{1}{\pi} d\, Area = \frac { dx\wedge dy}{\pi}$;

an {\it analytic Jordan curve} means a smooth Jordan curve which
admits a real analytic parametrization.

\section{Area preserving diffeomorphisms} This section contains some immediate implications of the
equation of  area preserving diffeomorphisms related to the
parametrization
of an analytic Jordan curve. Later on  we shall see that this,
apparently innocent, Jacobian identity plays an important role in
the study of  moving boundaries governed by elliptic growth.

\subsection{Fourier expansion} Consider the equation
\begin{equation}\label{AreaPres}
 {\Im} (\overline{z_t} z_q) = 1,\end{equation}
where $t \in [0, T]$ is a non-negative variable (usually identified
with time), while $q \in [0, 2 \pi]$ is the parameter along the
Jordan analytic curve $C_t$. Equation (\ref{AreaPres}) can be
interpreted as an area preserving property: that is the Jacobian of
the transformation
$$ (t,q) \mapsto (x,y),     \ \  {\rm where} \ \ \ z(t,q) = x+iy,$$
is equal to $1$.

To be more precise, for a fixed $t$, we assume that the
$2\pi$-periodic real analytic map
$$ z(t, \cdot) : [0,2\pi] \longrightarrow \mathbb C$$
is an embedding, and its range is denoted by $C_t$. We denote by
$D(t)$ the interior of the Jordan curve $C_t$.

In view of the smoothness hypothesis imposed on $z(t,q)$ we can
expand the function $z(t,q)$ in a Fourier series
\begin{equation}\label{Fourierseries}
 z(t,q) = \sum_{-\infty}^\infty  a_k(t) e^{ikq}.
 \end{equation}
 We will assume that the dependence $t \mapsto z(t,q)$ is $ C^1$.
 Also, we put $w = e^{iq}$, so that
 we can rewrite (2.2) as
 $$ z(t,w) = \sum_{-\infty}^\infty a_k(t) w^k.$$

 By the analyticity assumption, there exists $\epsilon, 0 < \epsilon <1$, so that the above Laurent series is
 convergent in the annulus $1-\epsilon < |w| < 1+\epsilon.$

 Due to the real analyticity of the map $z$, the Fourier series for $z$ and its derivatives
 are absolutely and
 uniformly convergent, whence
 $$ z_q (t,q) = i  \sum_{-\infty}^\infty  k a_k(t) e^{ikq},$$
 and
 $$ \overline{z_t}(t,q) =   \sum_{-\infty}^\infty  \overline{\dot{a}_n(t)} e^{-inq}.$$

 Thus, equation (\ref{AreaPres}) becomes
 $$ 1 = \Im (\overline{z_t} z_q)  = \frac{1}{2} \sum_{k,n} (\overline {\dot{a}_n} k a_k + \dot{a_k} n \overline{a_n})
 e^{i(k-n)q} =$$
 $$ \sum_{m=-\infty}^\infty [ \frac{1}{2} \sum_{n=-\infty}^\infty ((n+m) \overline{\dot{a}_n} a_{n+m} +
n  \dot{a}_{n+m} \overline{a_n})] e^{imq}.$$

By equating the coefficients we find
$$ 1 = \frac{1}{2} \sum_{n=-\infty}^\infty (n\overline{\dot{a}_n} a_{n} +
n  \dot{a}_{n} \overline{a_n}),$$ and
$$ 0 = \frac{1}{2} \sum_{n=-\infty}^\infty ((n+m) \overline{\dot{a}_n} a_{n+m} +
n  \dot{a}_{n+m} \overline{a_n}),$$ whenever $m \neq 0$.

Since
$$ {\rm Area} (D(t)) = \int_{C_t} \frac{\bar{z} dz}{2i} = \int_{|w|=1}
\overline{z(t,q)} \frac{\partial z(t, w)}{\partial w} \frac{dw}{2i}
= \int_0^{2\pi} \overline{z(t,q)} \frac{\partial z(t, q)}{\partial
q} \frac{dq}{2i},$$ we derive the following remarkable identity.

\begin{prop} Under the assumption (\ref{AreaPres}), the family of domains
bounded by the curves $z(t,q), \ 0 \leq q \leq 2\pi,$ satisfy
$$ \frac{d {\rm Area}(D(t))}{dt} = 2 \pi.$$
\end{prop}

By regarding $z$ now as a function of $t$ and the complex variable
$w = e^{iq}$, we obtain, along the curve $C_t$, the following:
$$ \frac{\partial z}{\partial q} = \frac{\partial z}{\partial w} \frac{\partial w}{\partial q} = z_w i w.$$

Denote, for the sake of simplicity, $z' = z_w$.  Then, the master equation (\ref{AreaPres}) becomes
$$ \Re (w z' \dot {\overline{z}}) = 1.$$
Since $z(t,\cdot)$ is analytic in the annulus $0< 1-\epsilon < |w| < 1+\epsilon$, the following result follows.

\begin{prop} Under the above assumptions
$$ w z'(t,w) \dot{z}^\sharp (t, 1/w) + \frac{1}{w} z'^\sharp(t,\frac{1}{w}) \dot{z}(t,w) = 2, \ \  1-\epsilon < |w| < 1+\epsilon.$$
\end{prop}

(As usual,  for a complex analytic function $h(w)$, we denote by
$$h^\sharp (w) = \overline{h(\overline{w})},$$
 obtained from $h$ by conjugating its Taylor coefficients.)

\subsection {Analytic parametrization.}\label{analytic-param} The most studied case of the Laplacian growth process requires an additional
analyticity assumption. We devote the present subsection to this
scenario.
Assume that for all $t \in [0, T]$ the negative Fourier coefficients
vanish, i.e.,

\begin{equation}\label{no_negative}
a_k(t) = 0, \ \ k<0.
\end{equation}
It is not difficult to see that this will hold the whole evolution, $t \in [0, T]$.  The
equation (\ref{no_negative}) simply means that for a fixed $t$ the function $z(t,w)$ extends
analytically to the unit disk $w \in \mathbb D$. Since $z(t,\cdot)$ is a homeomorphism from
the boundary $\mathbb T = \partial \mathbb D$ to the curve $\partial D(t)$, the argument
principle implies that
$$ z(t,\cdot) : \mathbb D \longrightarrow D(t)$$
is a conformal mapping. Let $w = \Psi(t,z)$ denote the inverse
conformal mapping. Since a linear transformation $z \mapsto \alpha z
+ \beta, |\alpha|=1,$ leaves the equation (\ref{AreaPres})
invariant, we can assume without loss of generality that $z(t,0)=0$
and that $\rho(t) = z'(t,0)>0.$ To distinguish this case from the
general case considered in the previous subsection we shall denote
$\partial D(t)$ by $\Gamma(t)$. The function $p(t,z) = \log
|\Psi(t,z)|$ is, up to a constant factor, the Green function of the
domain $D(t)$, with the source at $z=0$. This means that $p(t,.)$ is
the unique harmonic function in the punctured domain $D(t)\setminus
\{0\}$ having zero boundary values on $\Gamma(t)$ and such that
$p(t,z) - \log|z|$ is harmonic at $z=0$.

Moreover, the harmonic conjugate function $ \arg \Psi(t,z)$ is, up
to an additive constant, equal to $q(z), \ z \in \Gamma(t)$.

In other words, for a fixed value of the parameter $t$, we have:
$$ \nabla^2 p(t,\cdot) = 2 \pi \delta (\cdot),  \ \  {\rm in} \ \ D(t),$$
$$ p(t,\cdot)|_{\Gamma(t)} = 0,$$
and
$$ (q_y(t,\cdot), -q_x(t,\cdot))  = ( p_x(t,\cdot), p_y(t,\cdot)).$$
Consequently, the normal velocity of the boundary equals
$$ V = \frac{\partial q(t,\cdot)}{\partial \ell} = \frac{\partial p(t,\cdot)}{\partial n} .$$
So, by the area conservation property, we have
$$ V dz \wedge d\ell = dz \wedge dq.$$

These equations define a specific dynamics of planar boundaries
known as Laplacian growth. For recent  guides to the mathematics
and physics behind Laplacian growth we refer to the volume
\cite{Mineev2007} and the survey \cite{Mineev-Put-Teo}. We will
return to this case after discussing the geometry of the moving
boundaries.

 \subsection{The Schwarz function} An important tool for studying the changing geometry of the moving boundaries is the Schwarz
 function \cite{Davis}, \cite{Shapiro}.
 Up to the complex conjugation it is simply the (local)
 Schwarz reflection with respect to an analytic curve.

On the real analytic smooth boundary $\Gamma(t)$  of $D(t)$ we
introduce the Schwarz function
$$ \overline{z} = S(t,z),$$
where $S$ is analytic in the variable $z$. The domain of definition
for $S(t,.)$ is at least a tubular neighborhood  of $\Gamma(t)$,
although the function may possess analytic extensions to much larger
sets. For instance, the Schwarz function of a disk centered at $z=a$
and of radius $r$ is the rational function
$$ S(z) = \overline{a} + \frac{r^2}{z-a}.$$
If a polynomial $P(z,\overline{z})$ vanishes on $\Gamma(t)$, then,
necessarily, the associated Schwarz function satisfies the algebraic
equation
$$ P(z, S(t,z)) = 0, \ \ z \in \Gamma(t).$$

\begin{prop} \cite{Howison1992}\label{normalvelocity} The normal velocity of the boundary satisfies
$$ V = \frac{S_t}{2i \sqrt{S_z}},$$
with the proper choice of the branch of the square root, so that
$1/\sqrt{S_z}= dz/d\ell$ along $\Gamma(t)$.
\end{prop}

\begin{proof}
By taking derivatives with respect to $t$ we have
$$ \overline{z_t} = S_t + S_z z_t.$$
When restricted to the boundary curve,
$$S_z = \frac{ d \overline{z}}{d z}$$
is a complex number of modulus one.

Fix a single-valued branch of $\sqrt{S_z}$ along $\partial D(t)$.
This is always possible since $1/\sqrt{S_z}$ equals to the unit
tangent vector to $\partial D(t)$ and, hence, is single valued near
$\partial D(t)$. Then the above equation becomes
$$ \frac{-S_t}{\sqrt{S_z}} = \sqrt{S_z} z_t -
\frac{\overline{z_t}}{\sqrt{S_z}},$$ or, equivalently,
$$  \frac{S_t}{\sqrt{S_z}} = 2i \Im  \frac{\overline{z_t}}{\sqrt{S_z}}.$$
Since $$\frac{1}{\sqrt{S_z}} = \frac{dz}{d\ell}$$ is the unit
tangent vector, then
$$  \frac{S_t}{\sqrt{S_z}} = 2i V$$
\end{proof}

\begin{thm} There exists a multivalued analytic function
$W(t,z)$, defined in a neighborhood of $\partial D(t)$, with the
property
$$ S_t = \partial_z W,$$
and such that $ \Re W $ is constant along $\partial D(t)$.
\end{thm}

\begin{proof} From the above computations we find
$$ \overline{S_t} = -2i \overline{\sqrt{S_z}} V = \frac{-2i}{\sqrt{S_z}} V,$$
whence the vector $S_t$ is collinear with the complex conjugate of
the normal to $\Gamma(t)$. Moreover, rewriting the last equation in
the form
$$ \overline{S_t} = \overline{\partial} \overline{W}$$
we infer
$$ |S_t| n = \overline{S_t} = \overline{\partial} \overline{W} = \nabla (\Re W), $$
where $n$ is the normal to $\partial D(t)$. Hence, it follows that
the boundary, $\partial D(t)$ is a level set of $\Re W$.
\end{proof}

\subsection{Laplacian growth} In this subsection we merely illustrate  few classical observations
related to the consequences of the dynamics (\ref{AreaPres}) under
the analyticity assumption $a_k(t) = 0, \ \ k<0, \ t \in [0,T]$.
That is, we assume again that the parametrization $z(t,\cdot)$ of
the curve $\Gamma(t)$ analytically extends to the interior of the
unit disk $\mathbb D$ and will use intensively the Schwarz function
techniques.

By returning to the notations introduced in Section
\ref{analytic-param} we can identify the complex potential $W$ with
the multivalued function $\zeta(t,z) = p(t,z) + iq(t,z) = \log
\Psi(t,z)$ and then study the analytic extension
of the Schwarz function $S(t,z)$.

\begin{thm} \cite{Howison1992} For every $t$,
there exists a tubular neighborhood $U$ of $\Gamma(t)$, such that
\begin{equation}\label{S_t}
S_t(t,z) = 2 \zeta_z (t,z), \ \ z \in U.
\end{equation}
\end{thm}

Note that the function $\zeta$ is multivalued and analytic in the
punctured domain $D(t)\setminus \{0\}$. Its derivative $\zeta_z$ is
therefore meromorphic there with a simple pole at $z=0$ and the
residue equal to $1$.

\begin{proof}Indeed, according to Proposition \ref{normalvelocity}, we have along $\Gamma(t)$:
$$ \frac{S_t}{2i \sqrt{S_z}} = \frac{\partial p}{\partial n},$$
so
$$ S_t = \frac{2i}{dz/d\ell}  \frac{\partial p}{\partial n} = \frac{2i (\partial p/\partial n) d\ell}{dz} =$$
$$ 2 \frac{(i\partial q/\partial \ell) d\ell}{dz} = 2  \frac{\partial (p+iq)/\partial \ell) d\ell}{dz} =
2 \frac{\partial \zeta}{\partial z}.$$
\end{proof}

As simple as it looks, equation (\ref{S_t}) has surprising
consequences. In order to unveil them, we start with the known
Plemelj-Privalov-Sokhotsky formula applied for a fixed $t$ to the
function $S(t,\cdot)$. Let
$$ S_{\pm} (t,z) = \frac{1}{2\pi i} \int_{\Gamma(t)} \frac{ \overline{\sigma} d \sigma}{\sigma - z},   \ \ z \in D(t), \ \ {\rm respectively}\ \
z \in \mathbb C \setminus \overline{D(t)}.$$ Then
\begin{equation}\label{SPP}
S(t,z) = S_+(t,z) - S_-(t,z),\ \ \ \  z \in U,
\end{equation}
where $U$ denotes, as before, a neighborhood of $\Gamma(t)$.
Similarly we decompose the function $\zeta_z(t,z)$ and find
$$ (\zeta_z)_-(t,z) = \frac{-1}{z}.$$
From the uniqueness of the above decompositions and $(\ref{S_t})$,
we infer
$$ S_-(t,z)_t = \frac{-2}{z},$$
or, for $z \notin D(t)$,
$$ \frac{d}{dt} \frac{1}{2\pi i}  \int_{\Gamma(t)} \frac{ \overline{\sigma}d \sigma}{\sigma - z} =  \frac{-2}{ z},$$
that is
$$  \frac{d}{dt}  \int_{D(t)} \frac{ d A(\sigma)}{\sigma - z} = \frac{-2}{  z}.$$

By integrating against a polynomial $f$ along the circle $|z|=R$,
with $R$ sufficiently large, we find the following general identity.

\begin{prop} If the parametrization $z(t,w)$ of the boundary of the domain $D(t)$
extends analytically to the interior of the unit disk $\mathbb D :=
\{w: |w|<1\}$, then for every polynomial $f(z)$ the following
identity holds:
$$  \frac{d}{dt}  \int_{D(t)} f(z) d A(z) = 2 f(0).$$
\end{prop}

Equivalently, $S_{-}(t,z) = \frac{-2t}{z} + h(z)$, where $h(z)$ is
an analytic function in the neighborhood of $\mathbb C \setminus
\overline{D(t)}$, vanishing at infinity, and independent of $t$. A
simple application of Cauchy's formula now yields that there exists
a complex valued measure $\mu$ supported on a compact set $K \subset
{D(t)}$, independent of $t$ (as proved above) and such that
$$S_{-}(t,w) = \frac{-2t}{z} -\frac{1}{\pi} \int_K \frac{d\mu(\sigma)}{\sigma-z},\ \ \ z \in
\mathbb C \setminus \overline{D(t)}.$$

By repeating the above calculations we find
$$ \int_{D(t)} \frac{dA(\sigma)}{\sigma-z} = \frac{-2t}{z} -\frac{1}{\pi} \int_K \frac{d\mu(\sigma)}{\sigma-z}$$
and, consequently, the following quadrature identity follows.

\begin{cor} Under the same hypotheses as in Proposition $2.6$, for every polynomial $f$ we have
\begin{equation}\label{quadrature}
\int_{D(t)} f(z) dA(z) = 2t f(0) +  \frac{1}{\pi} \int_K f(z)
d\mu(z).
\end{equation}
\end{cor}

Simple examples show that the measure $\mu$ is not unique. If one
insists that the supporting set $K$ is "minimal", and the
representing measure $\mu$ is positive, then one can prove in most
interesting cases the uniqueness of $\mu$. The case of quadrature
domains $D(t)$, corresponding by definition to a positive finite
atomic measure $\mu$, is by far the best understood from the
constructive point of view. In this case the conformal mappings
$z(t,w)$ are rational. Examples, a discussion of the alluded
uniqueness and further details and references can be found in the
collection of articles \cite{QD}, cf. also  \cite{Gu1990},
cite{Shapiro}.

\section{Elliptic growth}

Guided by Laplacian growth as a prototype,  we introduce in this
section the elliptic growth phenomenon mentioned in the
Introduction. It is surprising to see that many features of
Laplacian growth persist and yet  sharp differences occur.  Let us
start the formulation in arbitrary dimension $d$ for a possibly
multiply connected domain $D(t)$ in ${\mathbb R}^d$ with many
sources, but later on we will focus on a homotopically trivial 2D
case with a single source at the origin in more detail.

Consider a family $D(t)$ of bounded domains in ${\mathbb R}^d$
with smooth analytic boundaries. Moreover, dependence of $D(t)$ on
$t$ is assumed to be real analytic (in the sense of a chosen
parametrization) as well.

Let $G$ be an open set containing as relative compact subsets all
$D(t), -1 < t <1$, and let $\lambda: G \longrightarrow \mathbf (0,
\infty)$ be a real analytic function. We consider the elliptic
(non-positive) differential operator
$$ L = \nabla \cdot \lambda \nabla = {\rm div} (\lambda \, {\rm grad}).$$
As we noted in  Section 1.1, when there is no danger for confusion
we shall omit the dot in the notation, and write, for example,
$\Delta = \nabla^2$.

The moving boundary problem with $N$ sources $s_k$ at ${{\bf x}_k
\in D(t)}$ is the
following:\\

{\it Given $D(0)$, find domains $D(t)$  satisfying the system of
equations}:

$$ L\,p = \sum_{k=1}^N s_k \delta({\bf x} - {\bf x}_k)   \ \ \  {\rm in }\ \ \ D(t),$$

$$p|_{\partial D(t)} = 0,$$

$$V = \lambda \,\partial_n p \ \   {\rm on}\ \ \partial D(t). \ \ \ $$

Note that the first two conditions simply assert that $p$ is the
linear combination of the Green functions for the operator $L$ of
the domain $D(t)$, with singularities at ${\bf x_k}$, while the
third condition determines the dynamics of the moving boundary.


\begin{thm}\label{quadrature} For every  function $\psi \in C^2(G)$ satisfying $L\psi = 0$, we
have
$$ \frac{d}{dt} \int_{D(t)} \,\psi \,{\rm d Vol} =  \sum_{k=1}^N s_k \psi({{\bf x}_k}).$$
\end{thm}

Here, dVol stands for Lebesgue measure on $D(t)$.



Since the constant function $\psi = \mathbf 1$ is annihilated by the
operator $L$, the above formula implies
$$ \frac{d}{dt} {\rm Vol}\,{D(t)} = \sum_{k=1}^N s_k, \ \ |t|<1.$$
Thus, in this moving boundary process, the volume is still
proportional to time.

\begin{proof} Let
$d\Gamma$ denote the
surface element on each connected component of the boundary of
$D(t)$. Then,
we have:

$$ \frac{d}{dt} \int_{D(t)} \psi d {\rm Vol} = \int_{\partial D(t)} \psi V
d\Gamma =  \int_{\partial D(t)} (\psi \lambda \partial_n p - p
\lambda
\partial_n \psi) d\Gamma =
$$
$$\int_{\partial D(t)} (\psi \lambda \nabla p - p \lambda \nabla
\psi)\cdot n\, d\Gamma=  \int_{D(t)} \nabla\cdot(\psi \lambda \nabla
p - p \lambda \nabla \psi) d {\rm Vol} =$$
$$ \int_{D(t)} [\psi L p- p L \psi] d  {\rm Vol} =  \sum_{k=1}^N s_k \psi({{\bf x}_k}).$$

\end{proof}

\begin{cor} In the case when the domains $D(t)$ are all homeomorphic to a ball and contain a single source $s_1>0$,
the moments
$$C(\psi) =  \int \psi \, d {\rm Vol}, \ \ \ L\,\psi =0,$$ determine  the domains $D(t)$ (locally
in $t$).\end{cor}

\begin{proof} Indeed, it is sufficient to consider a single moment $C(\mathbf 1)$.
As remarked earlier,
$$ \frac{d {\rm Vol}(D(t))}{dt} =  \int _{\partial D(t)} V d \Gamma = s_1>0$$
and the corollary follows, after observing that the family $D(t)$ is
increasing with respect to the  ordering by inclusion.
\end{proof}

{\bf Remark.}  Theorem \ref{quadrature} and the Corollary extend
word for word to more general elliptic operators
$$ L = {\rm div}\, (\Lambda \,{\rm grad}) - u({\bf x}), $$
where the matrix $\Lambda= (\lambda_{i,j}({\bf x})_{i,j=1,2})$ is
uniformly elliptic on the domain $G$ and all the coefficients
$\lambda_{i,j}, u({\bf x})$ are assumed to be real analytic in $G$
and $u \geq 0$. The only modification needed is that in the last
boundary condition in (1.1) where one ought to require $V =
\Lambda(\nabla p) \cdot n$. The existence of the Green function $p$
for such operators is well known \cite{PW, Evans}.  The fact that
$V>0$ on $\partial D(t)$ then follows from the maximum principle and
Hopf's lemma which hold for such operators cf. \cite{PW, Evans,
Reichelt}.

Assuming that the sources strengths $s_k(t)$ depend on time we then
obtain another notable corollary  of Theorem \ref{quadrature}. The
functionals $\int_{D(t)} \psi d\mu$ do not depend on $s_k(t)$, but
only on the value of the integral $\int_0^t\,s(t)\,dt$
\cite{Richardson1981,Varchenko, Gu1990}.

The functionals $\int_{D(t)} \psi d{\rm Vol}$ have a remarkable
potential theoretic interpretation \cite{Varchenko,
Mineev-Weinstein1993}. Indeed, imagine that a domain $D(t)$ is
occupied by matter with a unit density, which creates the potential
$\Phi$, governed by the Poisson's equation
$$L\Phi = \chi_{D(t)},$$
where $\chi_D$ is the characteristic function of the domain $D$. A
solution of the last equation is
$$\Phi({\bf x}) = \int_{D(t)}\,G_0({\bf x},{\bf y})\,d{\rm Vol(y)},$$
where $G_0({\bf x},{\bf y})$ is the fundamental solution for the
operator $L$. In important particular cases which are relevant for
physical applications (for instance, for the Helmholtz operator
$\Delta-1$) it is possible to expand $G_0({\bf x},{\bf y})$ into the
series
$$G_0({\bf x},{\bf y}) = \sum_n {\tilde \psi_n}({\bf x}) \psi_n({\bf y}),\ \  x\notin D(t), y \in D(t),$$
where $\{\psi_n\}$ and $\{\tilde{\psi}_n\}$ are bases of the null
space of $L$ in $D(t)$ and its complement in ${\mathbf R}^d$
respectively.  Then, assuming commutativity of summation and
integration, we obtain
$$\Phi({\bf x}) = \sum_n {\tilde \psi_n}({\bf x}) \int_{D(t)}\, \psi_n({\bf
y})\,d{\rm Vol(y)}.$$ Therefore,  we have obtained the functionals
introduced in the Theorem \ref{quadrature} as the coefficients of
the multi-pole expansion of the non-Newtonian potential given in a
far field.  We would like to add that the gradient of $\Phi$ is a
generalization of the Cauchy transform for the domain $D(t)$ - the
notion that was so useful in the Laplacian growth and
 the related field of quadrature domains.

In the case of elliptic growth with nonhomogeneous density mentioned
in the Introduction, one should modify the formulation given in the
beginning of this section by adding a positive space-dependent
factor, $\rho$, namely,
$$\lambda \rightarrow \rho \lambda;$$
under these assumptions we still encounter  an infinite set of
conservation laws similar to the previous case, when $\rho = 1$,
namely:
$$\frac{d}{dt} \int_{D(t)}\,\frac{\psi}{\rho} d{\rm Vol} =  \sum_{k=1}^N q_k \psi({{\bf x}_k}).$$
The elementary proof is not included here.  From the point of view
of  potential theory this corresponds to the case of occupation of
the domain $D(t)$ by non-uniform matter with density $1/\rho({\bf
x})$.  Some aspects of this case in 2D were discussed in
\cite{ZabrodinSigma}, also cf. \cite{Reichelt}. It is clear that the
inverse potential problem of recovery of $D(t)$ is considerably more
difficult in this situation.

\section{The conjugate function}

From now on we return to a planar case with a single source of
strength $s_1=2\pi$ and assume
 that all domains $D(t), \ 0 \leq  t \leq  T,$ are
{\it simply connected}. Define, using the notations from the
previous section, a (multivalued) {\it conjugate function} $ q \in
C^\omega(D(t))$ by
\begin{eqnarray}
 q_y & = &\lambda p_x \\
 q_x & = & -\lambda p_y.
\end{eqnarray}
Accordingly,
$$ \nabla \cdot \frac{1}{\lambda} \nabla q = 0, \ \ \ {\rm in} \ \
D(t).$$ In most computations below $t$ is fixed. However, we stress
that by its very definition, the function $q$ depends on $t$ also:
$q(z, \overline{z})= q(t; z, \overline{z}).$ We hope that omitting
$t$ in the notations of $q$ will not  confuse  the reader.

We have
$$ V = \lambda \partial_n p = \partial_{\ell}\, q.$$

Let $z(t,\ell)$ be the parametrization of the contour
$\Gamma(t)=\partial D(t)$ by an arc-length $\ell$.  Then a right
angle rotation of the unit tangent vector gives
$$ n = -i z_{\ell},$$

hence the normal component of the boundary velocity is
$$ V = z_t \cdot (-iz_{\ell}) = \Re
(z_t (\overline{-iz_\ell})) = \Im(\overline{z_t}z_{\ell}) =
\partial_{\ell}\, q,$$ where, as before,
subscripts stand for partial derivatives.

Since $\lambda$ is positive, then $ \partial_{\ell}\, q>0$ along
$\Gamma(t)$, therefore $q(t,.)$ can equally well parameterize the
boundary $\partial D(t)=\Gamma(t)$, in which case we can rewrite the
above equation as
$$ \Im (\overline{z_t} z_q) = 1,\ \ {\rm on} \ \ \partial D(t).$$

\begin{lem} The variation of $q$ along the curve $\Gamma(t)$ is equal to $2\pi$.
\end{lem}

\begin{proof} Indeed
$$ {\rm var}\  q|_{\Gamma(t)} = \int_{\Gamma(t)} \frac{\partial q}{\partial \ell} d \ell =
\int_{\Gamma(t)} \frac{\lambda \partial p}{\partial n} d \ell = $$
$$\int_{\Gamma(t)} \lambda \nabla p \cdot n d\ell =
 \int_{D(t)} \nabla(\lambda \nabla p)  d Area = 2\pi.$$
\end{proof}

Introducing the multivalued function

$$ \zeta(z,\overline{z}) = p+iq,$$we have
$$ i \lambda \overline{\partial} p =
\overline{\partial} q,$$ or, still,
\begin{equation}\label{Beltrami}
{\bar \partial} \overline{\zeta} = \frac{1+\lambda}{1-\lambda} {\bar
\partial
\zeta},
\end{equation}
which is a form of the Beltrami equation \cite{Vekua}.  In terms of
a new variable
$$ \omega = \sqrt \lambda p + i\frac{q}{\sqrt \lambda},$$
the equation (\ref{Beltrami}) takes the canonical Carleman form
\begin{equation}\label{Carleman}
{\bar \partial} {\omega} = ({\bar
\partial} \log \sqrt \lambda)\,\overline{ \omega}.
\end{equation}
An implicit solution of this equation is found from
$${\omega}(z,\bar z) = F(z) \exp{\int_{D(t)}\frac{{\bar
\partial} \log
(\sqrt{\lambda(\zeta, \bar \zeta)})}{\zeta - z}\, \frac{
\overline{\omega( \zeta, \overline{\zeta})}}{\omega(\zeta, \bar
\zeta)} \, dA(\zeta)},$$

where $F(z)$ is analytic in $D(t)$ (cf. \cite{Vekua} for more
details).

Thus the moving boundary problem of finding $D(t)$ can be
reformulated as a Dirichlet boundary value problem:\\

{\it Given the weight $\lambda$, find a function $\zeta(z,\bar z)$
(or, $\omega(z, \bar z))$ satisfying the Beltrami (or, Carleman)
equation  above and subject to the boundary condition $\Re
\zeta = 0$ (or, $\Re \omega = 0$)}.\\

\section{Elliptic growth of Schr\"odinger type}
The previous section was devoted to the elliptic growth of the
Beltrami type.  In this section, we will consider the elliptic
growth of Schr\"odinger type, which, as already mentioned in the
introduction, is related to the theory of the Schr\"odinger
operator. To be specific, we consider the problem:

{\it Given a domain $D(0)$ find domains $D(t)$ , satisfying the
system of equations}:
\begin{eqnarray}
L\,P &=& (\nabla^2 - u)P = \sum_{k=0}^N q_k \delta({\bf x} - {\bf x}_k)   \ \ \  {\rm in }\ \ \ D(t), \nonumber \\
P    &=& 0 \ \ \ \quad \,\,{\rm on}\,\,\partial D(t), \nonumber \\
V    &=& \partial_n P \ \ \ {\rm on}\,\,\partial D(t)\nonumber .
\end{eqnarray}
As has already been demonstrated, this problem has an infinite set
of conservation laws which are time derivatives of integrals of null
vectors of the operator $L$. These integrals have the
potential-theoretic interpretation discussed in Section 3.  Let us
pose the following question. Does there exist in this case a
function $Q$ ``conjugate'' w.r.t. to $P$ in the sense that $P$ and
$Q$ are connected via some generalized Cauchy-Riemann equations?
And, if such $Q$ exists, can it be used as a parametrization of the
moving contour $\partial D(t)$ similarly to how  it was used in the
case of the elliptic growth of the Beltrami type in the previous
section?

The answer to the first question is `yes', and to the second one -
`no'.  One can see this from the generalized Cauchy-Riemann
conditions which connect $p$ and $q$ in the Beltrami case. Namely,
\begin{eqnarray}
\lambda \partial_x p &=& \partial_y q, \nonumber \\
\lambda \partial_y p &=&-\partial_x q \nonumber.
\end{eqnarray}
These formulae suggest the substitution
$$p = \frac{P}{\sqrt \lambda},$$
$$q = \sqrt \lambda Q.$$

Thus, the new functions  $P$ and $Q$ are connected via the system of
linear euqations:
$$\partial_x P - P \partial_x (\log \sqrt \lambda)= \partial_y Q + \partial_y (\log \sqrt \lambda),$$
$$\partial_y P - P \partial_y (\log \sqrt \lambda)= -\partial_x Q - \partial_x (\log \sqrt \lambda).$$

Differentiating the first equation w.r.t. to $x$, the second one
w.r.t. to $y$, and then adding them, one obtains
$$(\nabla^2 - u)P = 0,$$
$$(\nabla^2 - v)Q = 0,$$
where
$$u= \frac{\nabla^2 (\lambda^{1/2})}{\lambda^{1/2}};$$
$$v= \frac{\nabla^2 (\lambda^{-1/2})}{\lambda^{-1/2}}.$$
This simple transformation known  as the ``removal of the first
derivative from linear differential equations of the second order''
and, also, closely related to supersymmetry in physics, joints
together the two major types of elliptic growth.

Now let us show that the function $Q$, unlike the function $q$
cannot, in general, provide a parametrization of the contour.
Indeed, it was shown in the previous section that $q$ can serve as a
parametrization since it is a monotonically increasing function of
the arc-length along the contour. Since $q = \sqrt \lambda Q$, it is
now clear that $Q$, generally speaking, is not monotone along the
interface because of the space dependent factor $\lambda^{-1/2}$.

As one can easily see, $\sqrt \lambda$ solves the same Schr\"odinger
equation as $P$, namely
$$(\nabla^2 - u)\sqrt \lambda = 0,$$
while the function $1/\sqrt \lambda$ solves the same Schr\"odinger
equation as $Q$, namely
$$(\nabla^2 - v)\frac{1}{\sqrt\lambda} = 0.$$

\section{An inverse problem} We address below the following natural question:\\

{\it  Is it possible to have the same "movie" $(t, \Gamma(t)), \ t
\in [0,T]$ governed by the elliptic growth dynamics
with different weights $\lambda$?}\\

By studying a particular example we shall demonstrate  that, indeed, such\\
non-uniqueness may take place. To fix the ideas, assume that the
elliptic growth dynamics, as specified above has the property that
the conformal map $z(t, \cdot)$ onto a neighborhood $U$ of
$\Gamma(t)$:
$$ z(t,\cdot) : \{ w; 1-\epsilon < |w| <1+\epsilon\} \longrightarrow U$$
extends analytically to the unit disk $|w|<1$. Then, as we saw
earlier, the equation (\ref{AreaPres}) once again governs the whole
evolution of $\Gamma(t)$ but this time the evolution of $\Gamma(t)$
is the Laplacian growth.  In particular, the parameter along each
boundary satisfies $q(t,z) = \Im \log w(t,z)$, that is $\nabla^2 q =
0$ at all points  except the isolated singularity. On the other hand
we have started with the assumption $\nabla \lambda^{-1} \nabla q =
0$. Hence
$$ \lambda^{-1} \nabla^2 q - \lambda^{-2} (\nabla \lambda)\cdot (\nabla q) = 0,$$
so
$$ (\nabla \lambda)\cdot (\nabla q) = 0.$$
The  function $p(t,z)$ also, by its very definition, satisfies
$$( \nabla p)\cdot (\nabla q) = 0.$$
Since $\nabla p$ never vanishes, at least in a neighborhood of
$\Gamma(t)$, a functional dependence
$$ \lambda(z) = f(t, p(t,z))$$
must occur. This simple computation gives rise to the following
counterexample.

Let us consider the simplest Laplacian growth process of expanding
concentric discs
$$ D(t) = D(0, \sqrt{2t}), \ \ t >0,$$
with associated functions
$$ p(t,z) = \log \frac{|z|}{\sqrt{2t}}, \ \ q(t,z) = \arg z.$$
Let $r = |z|$ and choose any positive, smooth radial function
$\lambda(r)$, e.g., $\lambda(r) = \exp r.$ Then, everywhere in
$\mathbb C \setminus \{0\}$, we find
$$ \nabla \lambda(r)^{-1} \nabla q = \lambda(r)^{-1} \Delta q - \lambda^{-2} \nabla \lambda \cdot \nabla q
= 0.$$ Now we can solve the problem with a rotationally symmetric
function $p(t, r)$:
$$ \nabla \lambda(r) \nabla p = 2\pi \delta(0),\ \ \ \ p(t, \sqrt{2t}) = 0,$$
so that $p,q$ are $L$-conjugate, where $L = \nabla \lambda \nabla$.
In particular, the normal velocity of the boundary is given by
$$ \frac{\partial q}{\partial \ell} = \lambda \frac{\partial p}{\partial n},$$
thus it is independent of the choice of $\lambda$. That is, the same
movie can be described as the Laplacian growth with $\lambda =1$ and
as elliptic growth with any other positive weight which is
rotationally invariant.

Note that a point on the boundary $\Gamma(t) = \partial
D(0,\sqrt{t})$ is parametrized as
$$ z(t,q) = \sqrt{2t} e^{iq},$$ so that the associated conformal map is
$$ z(t,w) = \sqrt{2t} w,\ \ \ |w|=1.$$   On the other side, if we choose
an initial contour to be different from the level set of the
function $\lambda$, the inverse problem may have a unique solution
modulo an arbitrary function of $\lambda$.  Obviously, the simplest
way to verify if the ``movie'' in the previous example is the
Laplacian growth, is to run the ``movie'' again but with a
non-circular initial configuration. Then the ambiguity of the
previous example should disappear.  But this way of removing
non-uniqueness will require two ``movies''. Thus, it is probably
correct to say that there a continuum of $\lambda$'s, which
correspond to the same ``movie'', and that there are two sources of
non-uniqueness: (i) any smooth function of $\lambda$ can replace the
original function $\lambda({\bf x})$ and (ii) any smooth function of
$p$ can be multiplied by the original $\lambda({\bf x})$ without a
change of the ``movie''.  This is an interesting feature that merits
further investigation. Here, we have simply pointed out (by
constructing an example) the non-uniqueness of elliptic operators
$L$ providing the same evolution (i.e., a ``movie'').

\section{Schwarz function in elliptic growth}
It is instructive to state the moving boundary problem in terms of
the Schwarz function $S(z)$.  Start with the equation for the
velocity of the moving boundary $V =
\partial_{\ell}q$. By Proposition 2.3 its left hand side equals
$S_t/(2i\sqrt{S_z})$, while the right hand side can be written as
$\partial_{\ell} (\zeta - p)/i = \partial_{\ell}\,\zeta/i$ since
$\partial_{\ell} p = 0.$ (Here, as before $\zeta = p+iq$.) Then,

$$\zeta_{\ell}/i = (z_{\ell}\partial_z \zeta + {\bar
z}_{\ell}\partial_{\bar z}\zeta)/i = (\partial_z \zeta + S_z
\partial_{\bar z} \zeta)/i\sqrt S_z.$$  Hence, in virtue of
the Beltrami equation $(4.3)$  and the identity
$$ \frac{S_t}{2i \sqrt{S_z}} = -i \zeta_\ell,$$ we obtain
the following evolution law for the Schwarz function.

\begin{prop}\label{S_t-again}Under the above assumptions,
\begin{equation}\label{Schw-EG}
S_t (t,z) = 2(\zeta_z+\frac{1-\lambda}{1+\lambda}
\overline{\zeta_z}S_z) = 2\partial_z \zeta(z,S(t,z))
\end{equation}
for all $z$ in a neighborhood $U$ of $\Gamma(t)$.
\end{prop}

Notice that the above formula is surprisingly similar to the
evolution law of $S(t,z)$ in the Laplacian growth process - cf.
Theorem 2.5.

\subsection{Dynamics of singularities of the Schwarz function}
The above proposition has interesting consequences.  We consider
first the dynamics of the singularities of the Schwarz function and
associated universal quadrature formulas that are preserved during
the elliptic growth.  We have seen in the section devoted to the
Laplacian growth that the poles, or more generally, the Cauchy
integral density of the Schwarz functions of the moving boundaries
are unchanged with the exception of the residue of the pole at $z=0$
which depends linearly on time $t$. Naturally, we expect that the
relations in the elliptic growth process are more complicated. The
present subsection collects some observations along these lines.

\begin{thm} Let $f(z)$ be a real analytic function defined in a neighborhood
of the closed domains $\overline{D(t)}$ which are moving according
to the elliptic growth law with an associated operator $L$. Then,
\begin{equation}\label{Lquadrature}
\frac{d}{dt} \int_{D(t)} f d Area = 2 \pi \tilde{f}(0),
\end{equation}
where $\tilde{f}$ solves the elliptic Dirichlet problem:
$$ Lu = 0 \ \ \ {\rm in} \ D(t),\ \ \ \ \ u|_{\Gamma(t)} = f.$$
\end{thm}

\begin{proof} Using the notations introduced in the earlier sections and
previous computations, we infer
$$ \frac{d}{dt} \int_{D(t)} f d Area = \int_{\Gamma(t)} f V  d\ell = \int_{\Gamma(t)} f \lambda \partial_n p d\ell =
2\pi \tilde{f}(0).$$
\end{proof}

If the boundaries $\Gamma(t)$ remain real analytic during the time
of growth as is tacitly assumed throughout this note we can, as
before, decompose the Schwarz function as follows:
$$ S(t,z) = S_+(t,z) - S_-(t,z), \ \ \ z \in \Gamma(t),$$
where
$$ S_-(t,z) = -\int_K \frac{d\mu(t,\sigma)}{\sigma-z},\ \ \ z \notin K.$$
Here, $K$ is a compact subset of $D(t)$ independent of $t$ and $\mu$
is a complex valued measure smoothly depending on $t$.  Then,
assuming in addition that $f$ is (complex) analytic in a
neighborhood of $\overline{D(t)}$, we find
$$ \frac{d}{dt} \int_{D(t)} f d Area = \frac{1}{2i} \frac{d}{dt} \int_{\Gamma(t)} f(z) S(z,t) dz =
$$ $$ \frac{1}{2 i} \frac{d}{dt} \int_{\Gamma(t)} f(z) S_-(z,t) dz  =
\pi \frac{d}{dt} \int_K [ \frac{1}{2\pi i} \int_{\Gamma(t)}
\frac{f(z) dz}{z-\sigma} ] d\mu(t,z) = \pi \int_K f(\sigma)
\frac{d}{dt}\{ d\mu (t,\sigma)\}.$$

This representation of the derivative of the average of an analytic
function becomes interesting in several particular cases. In
particular, as in the conservation law obtained by taking $f=1$, cf.
Corollary $2.7$ and Corollary $3.2$, it follows that
$$ \int_K  \frac{d}{dt}\{ d\mu (t,\sigma)\} = 2.$$
Another application is given by the following.

\begin{prop} Assume that for all times $t$ the singularities of the Schwarz function
contained in $D(t)$ are simple poles; i.e., we have in $D(t)$:
$$ S_-(t,z) = -\sum_{j=1}^N  \frac{a_j(t)}{b_j(t)-z} + \ {\rm (analytic \ remainder)},$$
with $b_j(t) \in D(t)$ for all $j$ and $t$. Then,
\begin{equation}\label{linearsystem}
\sum_{j=1}^N [a'_j(t) f(b_j(t)) + a_j(t)b'_j(t) f'(b_j(t))] = 2
\tilde{f}(0)
\end{equation}
for every analytic function $f$.
\end{prop}

\begin{proof} According to the computations in the proof of Theorem 7.2, we have
$$ \mu(t,\sigma) = \int \frac{ d [\sum_{j=1}^N a_j(t) \delta_{b_j(t)}](\sigma)}{z-\sigma},$$
hence
$$ \frac{d}{dt} \int_{D(t)} f d Area = \pi \sum_{j=1}^N [a'_j(t) f(b_j(t)) + a_j(t)b'_j(t) f'(b_j(t))]$$

and the statement follows from $(7.2)$.
\end{proof}

Now, by choosing $2N$ linearly independent analytic functions $f_k$
we can, based on (\ref{linearsystem}), form a linear system of
equations that at least in principle determines the dynamics of the
poles $b_j$ and residues $a_j$.

\section{Herglotz theorem and generating closed-form solutions}

The Herglotz theorem \cite{Herglotz} establishes a fascinating
one-to-one correspondence between the singularities of the Schwarz
function of the contour and the singularities of conformal maps from
a vicinity of the unit circle to a vicinity of the contour under
consideration.  Basically, it  states that if $a$ is a singular
point of a conformal map $f(w)$ from the unit disk $\mathbb D$ to
the domain $D$ such that $\partial D = \{ f(e^{iq}); q \in
[0,2\pi]\}$, then the Schwarz function of the curve $\partial D$ has
the singularity of the same kind at the point
\begin{equation}\label{Hergl1}
b = f(1/\bar a).
\end{equation}
Moreover, if an isolated singularity (i.e.,  a pole, an algebraic
singularity, or a logarithmic singularity) at $a$ appears in the
function $f(w)$ with a coefficient $A$, then the corresponding
singularity $b$ is present in the Schwarz function with the
coefficient$B$ determined from
\begin{equation}\label{Hergl2}
{\bar B} = A\,(-\bar a^2 f'(1/\bar a))^m,
\end{equation}
where $m$ is the multiplicity of a pole if $a$ is a pole; is a
rational number if $a$ is a an algebraic branch point; or,  if $a$
is a logarithmic singularity, $m$ is equal to zero.  Actually, the
last two equations, which we call Herglotz' theorem,  follow easily
from the representation of the Schwarz function $S(z)$ in terms of
the conformal map $f$ \cite{Davis},
$$S(z) = {\bar f}\circ(1/f^{-1})(z).$$
There is some evidence that  Herglotz' theorem should be  helpful in
solving the elliptic growth problem in terms of $z=f(t,w)$ and
generating exact solutions in the closed form.  Here, we present a
naive sketch of how we might expect  generating of exact solutions
for the elliptic growth should work.

First, find $\zeta = p+iq$ as a function of $z = x+iy$ and ${\bar z}
= x-iy$, either by solving the Beltrami equations (\ref{Beltrami})
or (\ref{Carleman}) for a given $\lambda$ and a given initial domain
$D(0)$, or by solving the Dirichlet problem in $D(0)$, thus finding
the Green function $p$ and consequently calculating the conjugate
function $q$ from the generalized Cauchy-Riemann equations (4.1),
(4.2). Then, the equation $p(z, \bar z) = 0$ implicitly defines the
Schwarz function of the moving boundary.

Second, either substitute $\zeta(z,S(z))$ into the equation
(\ref{Schw-EG}) for the dynamics of the Schwarz function $S(z)$
(scenario {\it A}), or, ``if one gets lucky'', try to invert the
solution $\zeta(z,\bar z)$, thus obtaining the function $z(\zeta,
\bar \zeta)$ (scenario {\it B}).\\

{\it Scenario} {\it A}: As the third step, identify  singularities
of $S(z)$ (and their dynamics) that are already built in
(\ref{Schw-EG}) through the function $\zeta(z,S(z))$ found at the
previous step and, also, identify the singularities of $S(z)$
which are {\it not} the singularities of the RHS of
(\ref{Schw-EG}). The latter are time-independent as one can see
from (\ref{Schw-EG}) and represent constants of motion associated
with the dynamics of the growth.

Scenario {\it A}. Fourth step. Using the formulae
(\ref{Hergl1}),(\ref{Hergl2}) provided by the Herglotz theorem
recover an {\it explicit} form for the moving boundary, $z = f(t,
e^{iq})$, with the time dynamics of all parameters of $f$ given
implicitly by (\ref{Hergl1}), (\ref{Hergl2}).\\

{\it Scenario} {\it B}. Third step. Restrict $z(\zeta, \bar
\zeta)$ to the imaginary axis of $\zeta$, that is the axis $p=0$,
thus defining the function $f(e^{iq}) = \zeta(iq, -iq)$, which, as
was shown above, satisfies  the equation (\ref{areapres}) for area
preserving diffeomorphisms.

Scenario {\it B}. Fourth step.  Substitute $f(e^{iq})$ into
(\ref{areapres}) as an initial condition with time-dependent
parameters and solve it. (The technique of integration is the same
as that for the Laplacian growth \cite{M-LANL}, but the
singularities of $f$ are now lie both inside and outside the unit
circle.) Alternatively, find the singularities of $S(z)$ that
correspond to the singularities of $f(z)$ through
(\ref{Hergl1}),(\ref{Hergl2}) and calculate their time dynamics
using (\ref{Schw-EG}).  In either case the Herglotz theorem plays
the central role in finding exact solutions for elliptic growth.


While a more complete theory for solving (\ref{areapres}) for the
elliptic growth exactly will be published elsewhere, we will present
below two examples of exact solutions in the case when the
singularities of $f(w)$ are poles, both inside and outside the unit
circle.

\subsection{Example: two simple moving poles}

Let us take
\begin{equation}\label{example1}
z(e^{iq}) = re^{iq} + \frac{A_1}{e^{iq} - a_1} + \frac{A_2}{e^{iq} -
a_2}
\end{equation}
as an initial condition for the equation (\ref{areapres}), assuming
$|a_1|<1$ and $|a_2|>1$.  One can verify by a direct substitution
that (\ref{example1}) is a solution of (\ref{areapres}) with time
dependent poles $a_{1}$, $a_{2}$, residues $A_{1}$, $A_{2}$ and the
conformal radius $r$.  The time dependence of these parameters is
given by the equations
\begin{eqnarray}\label{ex1}
b_{1,2} &=& z(1/{\bar a_{1,2}}),\\
B_{1,2} &=& - {\bar A_{1,2}}\,{\bar a_{1,2}^2}\,z'(1/{\bar a_{1,2}}),\\
2t + C  &=& r^2 - 2r\,\Re\frac{A_2}{a^2_2} -
\frac{|A_1|^2}{(1-|a_1|^2)^2} + \frac{|A_2|^2}{(1-|a_2|^2)^2},
\end{eqnarray}
where $b_{1,2},\,B_{1,2},$ and $C$ are constants of integration.
Actually, the RHS of the last equation is the area of the domain
$D(t)$ enclosed by the contour $z(t, e^{iq})$.  When $A_2=0$ this
is a standard Laplacian growth with a simple pole at $a_1(t)$.  In
this case, the contour is known  to develop a finite time
singularity by forming a cusp.  When $A_1=0$, the process becomes
a  so-called inverse Laplacian growth, when a more viscous fluid
displaces a less viscous one oppositely to a standard situation in
which  it is the other way around.  The inverse Laplacian growth
is stable and the shape rounds off during the evolution forming as
a rule a circle as a long time limit.  The case when both $A_1$
and $A_2$ are not equal to zero is a general case with a
nontrivial dynamics caused by an interlay between the
stabilization of a contour, due to the term in (\ref{example1})
with the pole $a_2$ that lies outside the unit circle, and
destabilization, due to the term in (\ref{example1}) with the pole
$a_1$ inside the unit circle.

It is interesting to note that the constants of integration in the
LHS of (\ref{ex1})-(8.6) describe singularities of the Schwarz
function of the moving contour, namely $b_{1,2}$ and $B_{1,2}$ are
the simple poles and residues of $S(z)$ respectively. Note that $b_1
\notin D(t)$ while $b_2 \in D(t)$.  Finally, $2t + C$ is the residue
of the Schwarz function at the simple pole at the origin, which
represents  the area of $D(t)$. It only remains now to find the
explicit expression of the function $\lambda({\bf x})$ for this
process described by (\ref{example1}).

\subsection{Example: two multiple stationary poles}

For the second example let us take, as the initial condition for the
(2.1), a function
\begin{equation}\label{example2}
z(e^{iq}) = re^{iq} + ae^{i(1-n)q} + be^{i(1+n)q},
\end{equation}
that describes a contour with $n$-fold symmetry which represents a
circle of radius $r$ modulated by a monochromatic wave in such a way
that exactly $n$ waves of an amplitude $\sqrt{2(|a|^2+|b|^2)}$ fit
the circumference (at least when $|a|$ and $|b|$ are both small with
respect to $r$).

One can substitute (\ref{example2}) into (2.1) and verify that
(\ref{example2}) is a solution of the area preserving diffeomorphism
(2.1) if the time dependent parameters $a$, $b$, and $r$  obey  the
following algebraic equations:
\begin{eqnarray}
A = a^{1/(1-n)}\,b^{1/(1+n)},\\
\frac{B}{r} = (1+n)a^{1/(n-1)} + (1-n)A a^{1/(1-n)},\\
2t + C = r^2 - (n-1)\,|a|^2 + (n+1)\,|b|^2,
\end{eqnarray}
where $A$, $B$, and $C$ are constants of integration.

The comments that can be made here about the dynamics of the
contour described by (\ref{example2}) are very similar to those
made above for the first example.  When $b=0$, the moving curve
blows up in a finite time by
forming cusps and thus ceases to exist after that.  When $a=0$ the
process is the inverse (stable) Laplacian growth, it smoothes the
curve to a circle in a long term asymptotics.  When both $a \neq
0$ and $b \neq 0$, this is an intermediate case which can, for
instance, stabilize the dynamics and prevent the finite time blow
up with a proper choice of initial parameters.


Just as in the previous example, the constants of integration
describe singularities of the Schwarz function, namely $A$ and $B$
are the $(1-n)^{th}$ and $(1+n)^{th}$ coefficients of the formal
Laurent expansion of S(z), while $2t + C$ is the area of $D(t)$
enclosed by $z(t,e^{iq})$.  Unfortunately, again, it is not clear at
the moment what elliptic parameter $\lambda$ is associated with this
process. This question requires further investigation.

\section{Conclusions}

In this work we have shown that the elliptic growth processes, which
present a natural generalization of the Laplacian growth, possess
some remarkable mathematical properties strongly resembling those
pertinent to the Laplacian growth.  Specifically, there is an
infinite set of conservation laws; these conservation laws can be
interpreted  in terms of potential theory; there exists a
parametrization of the moving interface by the stream function of an
associated fluid velocity vector field and, last but not least,
there are several interesting accompanying features related to the
singularities of the Schwarz functions of the moving boundaries. We
 expect the latter  to be especially helpful in generating and
illuminating particular closed form solutions of elliptic growth
problems. We think that the next step in this direction should be a
search for a class of multipliers $\lambda$, which will allow
explicit solutions of the Beltrami equations and, therefore, will
offer new closed form solutions of the elliptic growth phenomenon.




\end{document}